\documentclass[runningheads]{llncs}

\newcommand{\sumALL}[1]{\sum_{j_1=1}^{k_1}\sum_{j_2=1}^{k_2}\ldots\sum_{j_{i-1}=1}^{k_{i-1}}\sum_{j_{i+1}=1}^{k_{i+1}}\ldots\sum_{j_n=1}^{k_n} A^i_{j_1 j_2\ldots
j_{i-1}#1 j_{i+1}\dots j_n} x^1_{j_1} x^2_{j_2}\ldots
x^{i-1}_{j_{i-1}} x^{i+1}_{j_{i+1}} \ldots x^n_{j_n}}

\newcommand{\smallsumALL}[1]{\sum_{j_1=1}^{k_1}\sum_{j_2=1}^{k_2}\ldots\sum_{j_n=1}^{k_n} A^i_{j_1 j_2\ldots
j_n} x^1_{j_1} x^2_{j_2} \ldots x^n_{j_n}}

\usepackage{algorithmic,amssymb,amsmath,algorithm}
\usepackage{graphicx}
\usepackage{setspace}
\usepackage{url}
\urldef{\mailsa}\path|{ratnik_gandhi,
samaresh_chatterji}@daiict.ac.in|
\newcommand{\keywords}[1]{\par\addvspace\baselineskip
\noindent\keywordname\enspace\ignorespaces#1}
\begin{document}

\mainmatter  

\title{Some Algebraic Properties of a subclass of Finite Normal Form Games}

\titlerunning{SOME ALGEBRAIC PROPERTIES OF A SUBCLASS OF FINITE NORMAL FORM GAMES}

\author{Ratnik Gandhi \and Samaresh Chatterji}
\authorrunning{R GANDHI AND S CHATTERJI}

\institute{Dhirubhai Ambani Institute of Information and Communication Technology,\\
Post Bag No. 4, Gandhinagar 382007, India\\
\mailsa\\
\url{http://www.daiict.ac.in}} \toctitle{Some Algebraic Properties
Of a subclass of Finite Normal Form Games} \tocauthor{Ratnik Gandhi
and Samaresh Chatterji} \maketitle

\begin{abstract}
        We study the problem of computing all Nash equilibria of
        a subclass of finite normal form games. With algebraic characterization
        of the games, we present a
        method for computing all its Nash equilibria. Further, we present a method for deciding
        membership to the class of games with its related results.
        An appendix, containing an example to
        show working of each of the presented methods, concludes the
        work.
    \keywords{Galois group, Multivariate Newton-Raphson Method, Nash equilibria,
    Polynomial Algebra, Finite normal form games, Game Theory.}
\end{abstract}

\section{Introduction}
        \noindent Game theory has become an important technique for analyzing
        interactions amongst players (rational decision makers) in a
        competitive scenario, where each player tries to maximize his
        profit.  A fundamental concept in game theory is that of a Nash
        equilibrium in which every player is satisfied with his move. The
        concept was introduced by John Nash in his celebrated
        1951 paper \cite{jfn51}. Nash also proved that every mixed game has a Nash
        equilibrium. Nash's proof was a pure existence proof, and did not
        indicate any methods for computing Nash equilibria.\\

        \noindent In recent years, however, the problem of computing Nash
        equilibria has gained prominence, and has generated substantial
        research literature. In this work we consider methods for computing Nash
        equilibria of finite normal form games -- games known to have finitely many
        solutions\footnote{For the class of games that we consider, Harsanyi\cite{JCH73} shows that
        all the equilibria of the game are isolated and are odd in numbers.} -- that emphasize use of polynomial
        algebra.\\

        \noindent Nash equilibria of a game can be viewed as solutions to a system of equations
        and inequalities defined over payoffs and strategies. More specifically, they are the states of the game in which
        no player can obtain a more favourable outcome by a unilateral
        change of strategy. This system of inequalities can be converted into
        a system of polynomial equations that we call the \textit{game system} ($\mathcal{GS}$).
        We adopt this characterization of Nash equilibria and apply
        polynomial algebra as a computational framework.
        Note that the conversion of inequalities to equalities causes $\mathcal{GS}$ to
        have more solutions then just the Nash equilibria.
        Our
        main objective is to propose a method for computing all
        the Nash equilibria for a suitable  subclass of finite normal form
        games, that serves as an alternate to existing numerical or algebraic
        methods.\\

        \noindent Tight complexity bounds have recently been presented for the problem of computing an equilibrium
        for finite normal form games.\footnote{For $n\geq2$-player games \cite{CD06} and \cite{DGP06} show that
        problem of computing an equilibrium is PPAD-complete.}
        In the light of this and related results, it is of interest to focus on restricted classes of games and develop
        methods for computing their Nash equilibria.
        This is also of value in terms of applications of game theory
        in particular domains.\\

        \noindent Another aspect of our work relates to is that of finding all Nash
        equilibria rather than a single one.
        From an application perspective, this helps in better
        strategic decision making.\\

        \noindent Algorithms for computing all Nash equilibria,
        characterized as solutions of $\mathcal{GS}$, typically iterate the procedure
        for a single solution(sample solution).
        In their investigation of these algorithms,
        McKelvey and McLennan \cite{MM96} raised an important question:
        Whether a method can be found for computing all the
        equilibria of the input game, given a single equilibrium (referred to hereafter as a
        sample equilibrium), without repeating the solution
        procedure for the sample equilibrium.\\

        \noindent Motivated by the question raised by McKelvey and McLennan, and
        the tight complexity bound for computing an equilibrium, we
        consider the problem of computing \emph{all} Nash equilibria
        of \emph{subclass} of finite normal form games.
        The subclass of finite normal form games that we consider have
        all integral payoffs and
        all irrational equilibria. The class of games are called integer payoff
        irrational equilibria(IPIE) games. We develop algorithms for computing all its Nash
        equilibria using a sample solution,\footnote{A sample solution is a solution tuple of $\mathcal{GS}$
        with all irrational coordinates.} thus answering McKelvey and McLennan's question in the
        affirmative for games in this class. We further
        present an algorithm for deciding membership to the class of
        IPIE games. \\

        \noindent Our overall philosophy is to exploit the Galois group of univariate
        polynomial in $\mathcal{I}$ of $\mathcal{GS}$
        along with a single sample solution to extend our knowledge about
        the remaining solutions of the $\mathcal{GS}$, which include all the Nash
        equilibria. It is known that knowledge of the
        Galois group does not presuppose knowledge of all the roots.\footnote{
        A method for
        computing Galois group using Tschirnhaus
        Transformations is presented in \cite{GK00}.} In the
        remaining treatment, it is therefore usually assumed that the Galois
        group of irreducible univariate polynomials in the the ideal $\mathcal{I}$ of $\mathcal{GS}$ is known.\\

        \noindent The primary setting of our work remains the Galois theory over
        commutative rings \cite{chr65}.
        Accordingly, several subsidiary results of an essentially algebraic
        nature are derived in the course of our development. We also briefly
        consider the possibility of games over finite fields.\\

        \subsection{Related Work}

            \noindent To our knowledge, a method for computing Nash equilibria of IPIE
            games with Galois groups has not been considered
            earlier. An algorithm for fast decomposition of univariate
            polynomials, over the field of rational numbers,
            with known Galois group is presented by Enge and
            Morain \cite{EM03}. The
            algorithm decomposes univariate polynomials with Galois or non-Galois field
            extensions. Segal and Ward \cite{sw86} also considers the use of known
            Galois group in the problem of computing weight distributions in
            irreducible cyclic codes. The use of Galois group for computing
            roots of a polynomial is also mentioned in
            \cite{bmk71,Brodzik2006}.\\

            \noindent In our method, we consider systems of
            multivariate polynomials and their solutions. All
            the algebraic extensions we consider are extensions of the ring of
            integers $\mathbb{Z}$. Further, we are able to obtain results
            related to games with payoff values belonging to a discrete set, a
            finite set or a finite field. \\

            \noindent Existing methods for computing Nash equilibrium, such as the
            approach based on the Gr\"{o}bner bases are computationally
            inefficient. Homotopy continuation methods have the added drawback
            that they provide solutions via approximation.
            The method presented in \cite{BVV07} is highly
            dependent on the probability distributions chosen. Our method
            offers a reduction of computational time (compared to
            existing methods). Further, it computes
            exact equilibria for a subclass of IPIE games.\footnote{
            cf. Proposition \ref{result: intgame_closedfrm_equi}.}\\

           \noindent Section 2 outlines the
            underlying model for the class of games that we consider. Section
            3 presents the algorithm for computing all equilibria of
            IPIE games. Results related to the games are
            presented in Section 4. Issue
            relating to the computational complexity of the algorithm is discussed in Section 5. Section 6 outlines a method
            for deciding membership to the class of integer payoff
            games. Section 7 presents our conclusions and suggests directions for future work.
            An appendix, presenting detailed examples to show working of the methods, concludes the work.

        \section{Underlying Model}\label{section_model}
        {

            In this section we give the underlying model for the class of games we
            are going to consider.
            The required definitions have also been
            presented.

                \begin{definition}
                    A strategic finite normal form game $T_S$ is a 3-tuple $ \langle N,S_i,c_i : i \in N \rangle$, where,
                    $N$ is a non-empty finite set of players,
                    $S_i$ is a non-empty finite set of strategies
                    employed by player $i$ and for each player $i$ its payoff
                        $c_i$ is of the form
                        $c_i : \times_{k\in N} S_k \rightarrow \mathbb{R}$.
                \end{definition}
                \noindent Each player $i$'s mixed strategy $\Delta(S_i)$ is a probability
                distribution on its set of pure strategies $S_i$,
                i.e. from $S_i$ player $i$ chooses strategy $j$ with probability $x_j^i$,
                                where $x_j^i \in \Delta(S_i)$.
                A finite normal form game with mixed strategies and
                expected payoff $\alpha_i : \times_{k\in N} \Delta(S_k) \rightarrow
                \mathbb{R}$
                is called \textit{mixed extension} $T_M$ of strategic game $T_S$.\\

            \noindent Let $T_M$ be a finite normal form game with $n = |N|$ players. Each player
            $i$ has $k_i \geq 2$ strategies, $|S_i|=k_i$. We write
            $\mathcal{K}^+ = \sum_{i=1}^n k_i$ and $\mathcal{K}^* =\prod_{i=1}^n
            k_i$.
            $A^i_{j_1 j_2 \ldots j_n}$ denotes the payoff received by player $i$ when
            each player adopts strategy $j_m$ for $1\leq j_m\leq k_m$ and $m=1,\ldots,n$.
            The probability that player $i$ chooses strategy $j_i \in \{1,2,\ldots,k_i\}$
            is denoted by $x^i_{j_i}$,
            \begin{equation}\label{geq_zero_0}
                0 \leq x^i_{j_i} \leq 1.
            \end{equation}
            Moreover, for each player $i$,
            \begin{equation}\label{sum_one}
                \sum_{j_i=1}^{k_i} x^i_{j_i} = 1.
            \end{equation}
            \noindent Expected payoff for player $i$,
            \begin{eqnarray}\label{expected_payoff}
                \alpha_i  &=& \smallsumALL. \nonumber \\
           \end{eqnarray}
                    \begin{definition}
                        Given mixed extension of strategic game, mixed Nash
                        equilibrium is an action profile $\{x_{j_i}^i\}
                        \in \Delta(S_i)\ \forall i,j_i$ such that each player's mixed strategy maximizes
                        his payoff if the strategies of the other players are held fixed.
                    \end{definition}
            \noindent In a Nash equilibrium, the following holds:
            \begin{eqnarray}\label{nash_inequality}
                \alpha_i \geq \sumALL {j_i}, \nonumber\\
                \qquad \qquad \mbox{for every } j_i \in S_i \mbox{ and for every } i \in
                \{1,\ldots,n\}. \nonumber\\
            \end{eqnarray}
            The class of games that we consider can be defined as
            follows.
           \begin{definition}
                Finite normal form games with all integer payoffs and
                all irrational equilibria are called Integer Payoff
                Irrational Equilibria(IPIE) game $T$.
            \end{definition}
            It is clear that  $T \subset T_M$.
            Applying \eqref{geq_zero_0} and \eqref{sum_one} to
            \eqref{nash_inequality}, we obtain a system $\mathcal{GS}$ of
            polynomial equations, called \textit{game system},
            \begin{eqnarray}\label{game_polynomial_0}
                x^i_{j_i}(\alpha_i - \sumALL {j_i}) = 0, \nonumber\\
                \qquad \qquad \mbox{for every } j_i \in S_i \mbox{ and for every } i \in
                \{1,\ldots,n\}. \nonumber\\
            \end{eqnarray}
            \noindent The class of games that we consider has integer payoffs
            and so coefficient $A^i_{j_1 j_2 \ldots j_n} \in \mathbb{Z}$. All equilibria of these games are irrational
            and so \eqref{geq_zero_0} changes to
            \begin{equation}\label{geq_zero}
                0 < x^i_{j_i} < 1.
            \end{equation}
            This changes \eqref{game_polynomial_0} to
            \begin{eqnarray}\label{game_polynomial}
                \alpha_i - \sumALL {j_i} = 0, \nonumber\\
                \qquad \qquad \mbox{for every } j_i \in S_i \mbox{ and for every } i \in
                \{1,\ldots,n\}. \nonumber\\
            \end{eqnarray}
            \noindent For the problem of computing Nash equilibria of IPIE games,
            we characterize games with a $\mathcal{GS}$ of the form
            $\eqref{game_polynomial}$.
            For deciding membership to the class of IPIE games, we characterize
            games with a $\mathcal{GS}$ of the form $\eqref{game_polynomial_0}$.
            Our convention is to write totally mixed real-irrational Nash equilibria
            as irrational Nash equilibria.
            Note that all Nash equilibria of
            the game correspond to solutions of $\mathcal{GS}$, but
            the converse is not necessarily true.\\
            Next, we define some algebraic concepts required for our method.
            \begin{definition}
                Let $G$ be a group and $X$ be a set. Then an
                action of $G$ on $X$ is a function of form $G \times
                X \rightarrow X$.
            \end{definition}
            \noindent We shall be specifically interested in the
            following situation.
            \begin{definition}\label{def_gal_ext_ring}\cite{JLY02}
                Let $S$ be an extension of commutative ring with $R$, i.e.$R$ is a subring of
                $S$. Let $G$ be a finite group acting as
                $R$-algebra automorphisms on $S$. Then we define
                $S^G$ as the subring
                \[
                    S^G = \{s \in S | \forall \sigma \in G, \sigma s =
                    s\},
                \]
                and say that $S$ is a Galois extension with
                group $G$, if
                \begin{itemize}
                    \item $S^G = R$, and
                    \item for any maximal ideal $\mathfrak{m}$
                    in $S$ and any $\sigma \in G \backslash
                    \{1\}$, there is an $s \in S$ such that $\sigma s - s \notin
                    \mathfrak{m}$.
                \end{itemize}
            \end{definition}
            \noindent We now consider the particular situation when
            $S$ is an extension of $R$ of the form
            $R(\alpha)$ where $\alpha$ is a root of a
            polynomial $p(x) \in R[x]$. It is known that the
            Galois group $G$ of the extension $S$ acts as a
            permutation group on the roots of the polynomial
            $p(x)$.\\

            \noindent When Galois group acts on a subset of roots, due
            to group actions, we get the other elements of the set of
            roots.\footnote{Not necessarily distinct.}
            A transitive group action by a Galois group on an element of the root set
            produces all other elements of the root set.
            Transitivity  of a group can be formally defined as
            follows:
           \begin{definition}
                A subgroup $H \subset S_n$ is transitive if
                for every pair of elements $i,j \in
                \{1,2,\ldots,n\}$, there is $ \tau \in H$
                such that $\tau(i)=j$.
            \end{definition}
            \noindent A conjugate element of an element of a set
            is an element generated by a group action on the element.
            Set of all conjugate elements of an element is called an
            orbit of the element under the action of a particular
            group.
            Formally,
            \begin{definition}
                For every $x \in X$ we put $Gx= \{gx : \forall g \in G\}$,
                and call it the orbit of $x$ under $G$, or simply G-orbit of $x$.
            \end{definition}
            \noindent Orbit of an element under Galois group is called
            Galois-orbit of the element.
            If $G$ acts transitively on $X$ then there is only one G-orbit, $X$
            itself.\\

            \noindent For more game theoretic concepts refer \cite{or99},
            and for Galois theory related concepts refer Cox\cite{dc04}.

    }
\section{Method}\label{section:method}

        \noindent In this section we present outline of an
        algorithm for computing all Nash equilibria of IPIE games.
        The algorithm has two stages: a sample
        solution computation and the group action.
        Various methods for computing a sample
        solution are presented in \cite{MM96}.
        In this work we use a version of Multivariate Newton
        Raphson method(MVNRM).\\

        \noindent Over all approach of the algorithm is as follows:
        In the first stage compute a sample solution of the $\mathcal{GS}$.
        Once the sample solution is available, apply group action by Galois
        group to produce conjugate solutions of the sample solution. Rejecting
        all non-equilibria solutions from the set of solutions gives all the equilibria
        of an IPIE game.\\

        \noindent Note that the input to the algorithm is an IPIE game, and
         Nash \cite{jfn50} gives existence of at least one mixed strategy
         equilibrium(irrational equilibrium in our case). This implies that
         for each indeterminate variable, we are guaranteed to get a polynomial irreducible over
        the base ring $\mathbb Z$.\\

        \noindent In the first stage, MVNRM starts with a guessed solution of the $\mathcal{GS}$. For
        the class of games all the mixed strategy Nash equilibria (probability tuples)
        form a subset of the set of solutions of $\mathcal{GS}$. This allows
        choosing an initial guess of a solution to be either all 0's or
        1's or some value between $(0,1)$. The choice of a solution tuple speeds up the convergence rate of
        MVNRM. Next we convert the approximate solution of $\mathcal{GS}$ into algebraic form.\\

        \noindent For constructing minimal polynomials of approximate roots,
        output of MVNRM is fed as input to the KLL
        algorithm \cite{kll84}. The KLL algorithm requires
        at least $O(d^2+d\cdot\log H)$ bits of an approximate root for
        computing its minimal polynomial, where $d$ is the degree bound
        and $H$ is magnitude bound of the coefficients of the minimal polynomial.
        If we let $(x_{k1},x_{k2},\ldots,x_{kn})$ denote the
        approximate solution generated by MVNRM in $k^{th}$
        iteration, and $x_k =
        \sqrt{x_{k1}^2+x_{k2}^2+\ldots+x_{kn}^2}$,
        then the bound on number of bits required by the KLL algorithm sets the stopping criteria
        for MVNRM.

        \begin{proposition}\label{result: MVNRM stopping criteria}
            MVNRM must compute approximate solution
            of $\mathcal{GS}$ till the number of zero bits in the
            binary representation of $| x_{k+1} - x_k |$ is bounded above
            by $O(d^2+d\cdot\log H)$.
        \end{proposition}
        \begin{proof}
           MVNRM computes approximate roots of $\mathcal{GS}$, which form the input to the KLL
            algorithm. Required precision, $(d^2+d\cdot\log H)$, for the input to KLL algorithm
            then immediately sets the criteria for the amount of precision required by
            MVNRM.\hfill\qed
        \end{proof}

        \noindent Bounds on the value of $d$ and $H$ can be obtained as follows.
        The variety of polynomials in a Gr\"{o}bner basis and the polynomial system
        are same (Hilbert's Nullstalensatz). By change in a monomial order we can obtain a univariate
        polynomial in the desired indeterminate variable. It is easy to see that the bound on number of solutions
        of the system of polynomial equations also bounds the degree of each these univariate polynomials.\\

        \noindent Each polynomial in a Gr\"{o}bner basis of a polynomial system
        is essentially an S-polynomial. An S-polynomial is constructed by taking
        modulus of all the polynomials with a pair of
        distinct polynomials, for every pair of polynomials in the basis. Maximum value
        of coefficients in the $\mathcal{GS}$ then bounds the
        coefficient values of
        univariate polynomials in a Gr\"{o}bner basis of the $\mathcal{GS}$.
        A univariate polynomial in each indeterminate variable
        has a minimal polynomial as its one of the factors. And so the maximum magnitude
        of the coefficient in the minimal polynomial is bounded above by
        the maximum coefficient value of the game system $\mathcal{GS}$.\\

       \noindent At the end of the first stage,
        the algorithm generates a sample solution of the $\mathcal{GS}$.\\

        \noindent With the samples solution, available either in
        algebraic form or in numerical form, in the next stage of the
        algorithm, we apply group actions by Galois groups $G$.
        These groups are associated with the irreducible
        univariate polynomials in the ideal $\mathcal{I}$ of
        $\mathcal{GS}$. Note that, we assume Galois groups are known.
        For IPIE games the Galois groups are associated with ring extensions over
        $\mathbb{Z}$, and they generate conjugate solutions of the
        sample solution of the $\mathcal{GS}$.
            The group action is transitive and produces a single orbit
            for each indeterminant variable. Using all the orbits we can
            determine all the irrational solutions of the
            $\mathcal{GS}$.\\

        \noindent Recall that all the
        solutions of the $\mathcal{GS}$ need not be Nash equilibria. For rejecting
        unwanted non-equilibrium solutions, we apply the Nash equilibrium
        verification algorithm.
       \numberwithin{algorithm}{section}
            \begin{algorithm}[H]
            \caption{Phase 1: Computing a sample solution of $\mathcal{GS}$ with MVNRM and KLL Algorithm}
            \begin{algorithmic}[1]\label{algorithm: integer game}
                \STATE $X' = \{\}$. // Initialize an empty set to store a sample solution
                of the $\mathcal{GS}$.
                \STATE Characterize all the Nash Equilibria of the input game as solutions
                 to the system of form $\mathcal{GS}$.
                \WHILE{ one sample solution of polynomial system $\mathcal{GS}$ is not constructed}
                    \STATE Apply MVNRM with a starting solution tuple $x_0$ consisting entirely of
                    zeros.
                    \WHILE {inequality in Proposition \ref{result: MVNRM stopping
                    criteria} holds true.}
                        \STATE Compute approximate solution of $\mathcal{GS}$.
                    \ENDWHILE
                        \STATE Apply KLL Algorithm and compute
                        minimal polynomial.
                        \IF {minimal polynomials of some of the indeterminate variables in the approximate solution tuple is reducible over $\mathbb{Q}$}
                            \STATE Divide $\mathcal{GS}$ with those factors and go to step Step 4 with the updated game system.
                        \ELSE
                            \STATE Save the solution tuple in $X'$.
                        \ENDIF
                \ENDWHILE\\
            \end{algorithmic}
            \end{algorithm}

                \noindent The following algorithm
                computes transitive Galois group action of every
                indeterminate on its corresponding root in
                $X'$. It then generates Galois-conjugates of roots and saves them in solution tuple at
                appropriate position.

        \numberwithin{algorithm}{section}
             \begin{algorithm}[h]
            \caption{Phase 2: Computing polynomial time Galois group action on the
            sample solution in $X'$ to generate solutions of
            $\mathcal{GS}$.}
            \begin{algorithmic}[1]\label{algorithm: group action}
                \STATE Initialize processed-element list $X$ and unprocessed-elements list $U$ as $X = U = X'$.
                \WHILE {$U$ is not empty}
                    \STATE Let $u$ be the first element of $U$. Delete $u$ from $U$.
                    \FOR { each element $g$ in Galois group $G$}
                        \STATE Compute transitive Galois group action $u ^ g$.
                        \IF {$u ^ g \notin X$ }
                            \STATE $X = X \cup \{u ^ g\}$  and $U = U \cup \{u ^ g\}$.
                        \ENDIF
                    \ENDFOR
                \ENDWHILE
            \end{algorithmic}
            \end{algorithm}

            \noindent For the division by linear factor in Step 10,
            of Algorithm \ref{algorithm: integer game} we use multivariate polynomial division algorithm
            from \cite{GG03}.
       It is important to note that due to the existence of a mixed Nash equilibrium
        and the fact that all equilibria are irrational for the input game, we are guaranteed to get one
        solution of $\mathcal{GS}$ in $X'$ and so Algorithm \ref{algorithm: integer game} reaches Step 14 every time.
        Moreover, the finiteness of the group and the variety on which it
        acts guarantees
       that the Algorithm \ref{algorithm: group action} reaches Step 10.
        At the end of Step 10, Algorithm \ref{algorithm: group action} generates
        solutions of the system
        $\mathcal{GS}$ in $X$, all of which may not be Nash equilibria.
        At last, for rejecting non-equilibrium solutions, the polynomial time
        Nash equilibrium verification algorithm
        \cite{rg05} is used.\\

        \noindent Note that the method above
        computes solutions to a system of polynomial equation using its sample solution
        and its Galois group. After the computation of a sample solution, all other
        solutions computed are without factorization of the system of polynomial equations.

        \section{Results}\label{section:resultALLEqui}
       \noindent Observe that Step 9 in the Algorithm \ref{algorithm: integer game} rejects not only
        integer roots but rational roots. This forces
        the polynomial system not to consider rational extensions over
        the ring of integers. Such rejections can be justified
        with the following result.

        \begin{proposition}\label{prop_intgame_ratequi}
            Algorithm \ref{algorithm: integer game} and Algorithm \ref{algorithm: group action}
            cannot be used for computing all Nash equilibria of games
            with integral payoffs and rational equilibria.
        \end{proposition}
        \begin{proof}

            \noindent Let $T$ be a game with integer payoffs and one
            or more rational equilibria of form $a/b$, where $a, b \neq 0 \in \mathbb{Z}$.
            This forces an extension $S = \mathbb{Z}(a/b)$ over
            $\mathbb{Z}$. The group $G$ of automorphisms of $S$ which fix
            $\mathbb{Z}$ can be computed as follows.\\

            \noindent Let $c,d \in \mathbb{Z}$, for any $c+ (a/b)d \in S$ and for any $\sigma \neq id \in G$,
            \begin{eqnarray*}
                \sigma(c+ (a/b)d) &=& \sigma(c)+\sigma(a/b)\sigma(d)\\
                                  &=&    c     +\sigma(a/b) d,
            \end{eqnarray*}
            and $\sigma(\frac{a}{b}\cdot b) = a \Rightarrow \sigma(\frac{a}{b})\sigma(b) = a \Rightarrow \sigma(\frac{a}{b}) =
            a/b \Rightarrow \sigma =$ identity.

            \noindent This means, the group of automorphisms of rational extensions of the ring of integers turns out to
            be a trivial identity group. And so, the group doesn't
            provide necessary information for producing conjugate solutions of the $\mathcal{GS}$.\hfill\qed
        \end{proof}
        \noindent Now, in order to prove the validity of the proposed method for IPIE
        games, we establish 3 more results, as follows:
       \begin{proposition}
            For any IPIE game its Galois group is non-trivial.
        \end{proposition}
        \begin{proof}
            Any IPIE game, by definition, produce irrational ring extensions
            over the ring of integers $\mathbb{Z}$.
            We first show that each such
            extension produces a Galois extension.
            Let $\alpha =
            \sqrt[m]{n}$, for $m,n \geq 2$, be an
            irrational number. Let $p \neq n^m$ be a positive prime integer.\footnote{It is easy to see that if $p= n^m$ then
            any element $\beta \in p\mathbb{Z}(\alpha)$ is not a prime element.}
            Then it can be verified that $p\mathbb{Z}(\alpha)$ forms
            a prime and thus a maximal ideal of
            $\mathbb{Z}$.
            We now choose a $s \in
            \mathbb{Z}(\alpha)$ such that its co-prime to $p$.
            Following Definition \ref{def_gal_ext_ring}, if we let $S=\mathbb{Z}(\alpha)$
            and $R=\mathbb{Z}$, then for any $\rho \in G(S/R)\backslash\{1\}$ and for chosen $s$,
            $\rho s - s \notin p\mathbb{Z}(\alpha)$. This shows that
            the extension is indeed a Galois extension.\\

            \noindent Next, suppose the Galois group for
            the irrational extension $\mathbb{Z}(\alpha)$ is trivial.
            Then the minimal polynomial of $\sqrt[m]{n}$ has all
            factors linear over $\mathbb{Z}$, and hence $\alpha \in
            \mathbb{Z}$. This is impossible for IPIE games. And so
            the result follows.\hfill\qed
        \end{proof}
        \noindent The next result sets the criteria for the MVNRM to
        converge to a solution of a $\mathcal{GS}$.
        \begin{proposition}\label{result: convergence_criteria}
            Let $x=(x_1,x_2,\ldots,x_n)$  be the vector strategy tuple with each $x_i$
            denoting a pure strategy for the players and let
            $f(x)=(f_1(x),f_2(x),\ldots,f_n(x))$, for all polynomials $f_i \in
            \mathcal{GS}$. Then MVNRM converges to a sample solution of $\mathcal{GS}$ if
            the following condition holds: $|f(x)\ J^2(f(x))| < |J(f(x))^2|$.
        \end{proposition}
        \begin{proof}
            In MVNRM, an approximation of the $n^{th}$ strategy tuple $x_n$ is
            computed using
            \[
                x_n = x_{n-1} - \frac{f(x_{n-1})}{Jf(x_{n-1})}.
            \]
            If we let
            \begin{equation} \label{eqn: convergence_phi}
                \phi(x)= x - \frac{f(x)}{Jf(x)}
            \end{equation}
            then for overall convergence of MVNRM we need
            $|\frac{d}{dx}\phi(x)|<1$. Taking the derivative of
            \eqref{eqn: convergence_phi} and
            simplifying it, we get $|f(x)\ J^2(f(x))| < |J(f(x))^2|$.
            With this condition, MVNRM converges to a sample
            solution of the $\mathcal{GS}$.\hfill\qed
        \end{proof}
        With the required tools in hand, we can now show the correctness of
        the method for computing all Nash equilibria of IPIE games.
        \begin{proposition}
            Algorithms \ref{algorithm: integer game} and \ref{algorithm: group action}
             for computing all equilibria of IPIE games works. i.e. these algorithms generate all irrational
            equilibria of IPIE game at termination - and they do not generate any spurious solutions.
        \end{proposition}
        \begin{proof}
            The input to the Algorithm \ref{algorithm: integer game} is an IPIE game $T$ with $n$ players.
            All the Nash equilibria of this game are characterized by a polynomial system $\mathcal{GS}$
            of the form \eqref{game_polynomial}.
            The polynomial system comes from the inequalities
            on expected payoffs and payoffs at pure
            strategies. These inequalities cause the system to have
            more solutions then just equilibria.\\

            \noindent Algorithm \ref{algorithm: integer game} computes a sample solution
            of the $\mathcal{GS}$ using MVNRM and saves it in $X'$. MVNRM computes an approximate solution of
            the game system $\mathcal{GS}$, starting with an approximate solution value either 0,1 or $\frac{1}{2}$.
            Nash \cite{jfn51} guarantees that a solution of $\mathcal{GS}$ exists in $(0,1)$.
            With the convergence criteria from Proposition \ref{result:
            convergence_criteria}, MVNRM converges to a solution
            tuple of $\mathcal{GS}$. This may not be a sample
            solution.
            It is known that the input game has only irrational equilibria.
            If there are integer or rational solutions, they are rejected with
            the condition in Step 9 of the Algorithm \ref{algorithm: integer game}. This guarantees a sample
            solution of $\mathcal{GS}$.
            The polynomial bound given in Proposition \ref{result: MVNRM stopping criteria},
            on the amount of precision required for KLL algorithm, forces MVNRM to
            stop. KLL algorithm computes the minimal polynomial for each root in the sample
            solution and terminates. \\

            \noindent Roots in the sample solution extend the ring of
            integers $\mathbb{Z}$ to some Galois extension $S$ of
            it. This is due to the Galois correspondence established by
            Chase et al. \cite{chr65}. More so, irreducible polynomials of univariate
            polynomials in ideal $\mathcal{I}$ of $\mathcal{GS}$ forces transitivity of
            Galois groups. This gives us
            a meaningful non-trivial transitive Galois group
            $G = Gal(S/\mathbb{Z})$.\\

            \noindent The transitive group action has only one orbit.
            Polynomials that we consider are irreducible after elimination of
            integer and rational roots.
            Due to our assumption that the Galois group
            of irreducible polynomials are known, we get
            group action of
            a transitive group $G=Gal(S/\mathbb{Z})$ on elements of
            set $X'$ to generate all the elements in set $X$.
            System \eqref{game_polynomial} has finitely many real solutions
            and so the group action terminates. This enables Algorithm
            \ref{algorithm: group action}
            to reach Step 10. And so the algorithm
            generates all conjugate solutions of the sample solution containing all
            Nash equilibria of game $T$.\\

            \noindent The solutions of $\mathcal{GS}$ may be more than just the Nash equilibria.
            Unwanted solutions are further rejected by calling
            a polynomial time algorithm to verify Nash equilibrium.
            What remains are all irrational equilibria of the game $T$.
            And so the result follows.\hfill\qed
        \end{proof}

        \noindent
        With the algebra and algorithms discussed above, we further ask,
        whether we can consider finite normal form games with all irrational equilibria and
        all payoff values from either finite fields or some finite set.
        This question can be partially answered as follows:
        \begin{proposition}
            The algebra and algorithms for IPIE games cannot be extended to
            work over finite fields and their extensions.
        \end{proposition}
        \begin{proof}
            \noindent If we define a finite normal form game over some finite number field,
            then the only polynomial algebra that we can consider is
            congruent-modulo algebra. i.e. polynomial system of form
            \eqref{game_polynomial} will be modulo some prime or prime power.
            This forces the expected cost
            function codomain values to be restricted to the finite number field.
            The payoff functions in games must provide every player
            a choice over his strategies by suggesting an order between elements
            in the codomain, where the function maps strategies.
            It is known that, finite number fields are not
            ordered fields and so they fail to provide a total order amongst player
            strategies. Moreover, the available order over finite fields conflict with
            field operations and we cannot perform polynomial
            algebra. So, we cannot meaningfully define games, and
            consider polynomial algebra such as suggested in the
            Algorithms in Section \ref{section:method} for computing Nash equilibria of such games.\hfill\qed
        \end{proof}

            \noindent However, we can consider some discrete codomain of the payoff
            function
            such as the ring of integers with total order or some finite
            set with arbitrarily defined total order and consider games over
            it.\\

        \noindent Now, we consider a result that talks about computation of Nash equilibria in closed form.
        It is known that if a polynomial defined over fields has a solvable Galois group, then
        all its roots can be computed with radicals. If the result generalizes over rings then
        we can generalize the solvability by radical result.
                i.e. for some ring $S$ and a subring $R$
                the following holds:
                \begin{equation}
                    R = \mathbb{Z} = L_0 \subset L_1 \subset \ldots
                    \subset L_n = S,
                \end{equation}
                and $\exists\ \alpha_i \in L_{i+1}$, a
                natural number $n_i$, such that $L_{i+1} = L_i(a_i)$
                and $\alpha_i^{n_i} \in L_i$,
                then solvability by radicals can
                be extended for a subclass of IPIE games. All finite ring
                extensions need not be radical.
                With this
                restriction on the extension of the ring and the
                definition of Galois theory over rings, we have
                following result.
                \begin{proposition}\label{result:
                intgame_closedfrm_equi}
                    If the ring extensions associated with the IPIE
                    games are radical then all the equilibria of
                    IPIE games can be computed in closed form.
                \end{proposition}
                \begin{proof}
                    Follows immediately from the discussion
                    above.\hfill\qed
                \end{proof}

    \section{Computational Complexity}\label{section_complexity}
    {

                \noindent The characterization of equilibria as solutions to a system of polynomial equation is a
                polynomial time operation in the size of input
                payoff matrix, where the size of the matrix is $\mathcal{K}^*$.
                The while loop in side the Algorithm \ref{algorithm: integer game} of Steps (3-14) runs until a
                sample solution of the $\mathcal{GS}$ is computed. For $i\in\{1,\ldots,\mathcal{K}^+\}$
                and for each indeterminate variable $x_i$, let $d_i$
                denote the degree of its univariate polynomial in $\mathcal{I}$ of the $\mathcal{GS}$.
                The while loop of steps (3-14) runs for at most $d = \max_{i} d_i$
                times. Average case running time analysis of the Newton's method
                -- for computing approximation of all the roots of a univariate
                polynomial --
                is studied by Smale \cite{ss81,ss97}.
                A sufficient number of the steps for the Newton's method to obtain an
                approximate zero of a polynomial $f$, are polynomially bounded by the
                degree $d_i$ of the polynomial and $1/\rho$, where $\rho \in (0,1)$ is the probability
                that the method fails. Kuhn's algorithm improves
                efficiency by a polynomial factor and provides
                global convergence.  On the other hand Renegar
                \cite{jr87} studies the problem of computing approximate
                solutions of multivariate system of equations using
                homotopy method and presents an efficient algorithm.
                Note that these results on the complexity analysis
                has underlying assumption that the numerical method
                converges.\\

                \noindent For the Algorithm \ref{algorithm: integer game},
                number of operations, for constructing a minimal polynomial and
                checking its irreducibility over $\mathbb{Q}$,
                are bounded by a polynomial in the size of degree $d$ and maximum norm $H$
                of the minimal polynomial \cite{GG03}. The operation
                of multivariate polynomial division of Step 10 runs
                in $O(M_1\cdot M_2)$ time \cite{MP07}, where each $M_i$ is
                the maximum number of terms in the polynomials considered for
                division. Without loss of generality we let $M = M1 > M2$.\\

                \noindent With these details, we present
                the following complexity bound for computing a sample
                solution with the Algorithm \ref{algorithm: integer game}.

                \begin{proposition}
                    Algorithm \ref{algorithm: integer game} runs polynomial in $O(\mathcal{K}^+d(1/\rho+H+dH+cM))$.
                \end{proposition}
                \begin{proof}
                    The while loop of (3-14) runs for at most $d$
                    times. Considering the complexity of computing
                    an approximate root of each univariate
                    polynomial with Newton-Raphson's method, the
                    MVNRM with Proposition \ref{result:
                    convergence_criteria} runs polynomial in
                    $O(\mathcal{K}^+d\cdot1/\rho)$. The KLL Algorithm
                    runs in $O(dH)$, requiring at most $\mathcal{K}^+$ repetition in worst case.
                    The operation of checking irreducibility of a minimal polynomial, in
                    worst case, is required for each indeterminate
                    variable and for every factor of the univariate polynomials. The
                    irreducibility check runs polynomial in $O(dH)$.
                    Finally, the multivariate division is called with
                    only a single root of the univariate
                    polynomial, making $M_2$ a constant
                    $c$. The division algorithm is also required for
                    $\mathcal{K}^+d$ times. Summing up all and
                    rearranging terms we get the result.\hfill\qed
                \end{proof}
                \noindent We are not considering the issue of computing the Galois
                groups in this work, i.e. we consider that the Galois groups are known.
                But to make the discussion complete, we give below,
                the complexity of computing a Galois group of the given polynomial.
                Computation of a Galois group requires polynomial time
                in the size of the input polynomial and the order of
                its Galois group. If $f(x)$ has degree $d$ then its
                Galois group can have at most $d!$ elements and so
                in worst case the computation takes exponential time. This is best
                know upper bound due to Landau \cite{sl84}.
                Lanstra \cite{hwl92} surveys computational complexity
                result of computing Galois groups and other related
                problems. \\

                \noindent Once a Galois group $G$ is completely known, we must find
                the Galois-orbit $Gx$ of every known root $x$ of indeterminate variable in the $\mathcal{GS}$.
                An orbit construction takes polynomial time with algorithm suggested by Luks \cite{el93}.
                In the worst case, the
                algorithm requires action of each of the Galois group generator $g' \in G' \subseteq G$ to
                each element of the set of roots. This gives worst case time
                $O(|G'|\cdot |X|)$.
                Finally, the operation of verification of a Nash
                equilibrium solution, runs polynomial
                in the size of total number of strategies $\mathcal{K}^+$.
    }
    \section{Membership}\label{section_membership}
    {
            \noindent With the characterization of
            games as game system $\mathcal{GS}$, in this
            section we outline a method for deciding membership to
            the class of IPIE games.

        \subsection{Method}
        \noindent The games that we consider are known to have integer
        payoff values
        and all irrational equilibria. The irrational equilibrium solutions induce
        irrational ring extensions. With this property, an intuitive and naive approach to
        answer the problem of deciding membership is as follows.
        Given an input game, we characterize all its equilibria as solutions to
        the system of polynomial equations of form $\mathcal{GS}$ in \eqref{game_polynomial_0}.
        After the characterization
        we ask, whether for each indeterminate variable its univariate
        polynomial has linear factors over the field of rational numbers or not.\footnote{
        It is important to note here that the membership to the class of IPIE games can be
        decided by considering irreducibility of the univariate polynomial over $\mathbb{Q}$.
        Justification of this fact comes from the following: if a polynomial
        has all its factor linear over $\mathbb{Q}$, then its roots are either integers or
        rationals. If the polynomial has no linear factors
        over $\mathbb{Q}$ then the roots are either irrational or complex.
        By \cite{jfn51} every game has a mixed strategy Nash equilibrium and so
        strictly a real number. This means that there is at least one irrational
        root of the polynomial.}
        If the polynomial has  no linear factors
        over $\mathbb{Q}$, then we must verify whether the
        solutions -- consisting of linear factors of the polynomial over $\mathbb{Q}$ --
        are Nash equilibria of the input game or not. If any of the solutions is an equilibrium then
        the game is a non-member to IPIE games. Otherwise it is.
        For the polynomial irreducibility test over $\mathbb{Q}$, we make use of
        a polynomial time univariate factorization
        algorithm over $\mathbb{Q}$ from \cite{GG03}.
        For constructing a univariate polynomial from multivariate system of polynomial equation
        $\mathcal{GS}$ we use a Gr\"{o}bner basis of the
        $\mathcal{GS}$.\\

        \noindent We start the membership decision with checking payoff values of the input
        game. If the payoff values are non-integer then we declare
        the input game a non-member to the class of IPIE games. In
        the case otherwise we do the following.

        \numberwithin{algorithm}{section}
        \begin{algorithm}[h]\label{algorithm:membership}
            \caption{\small Algorithm for deciding membership to the class of IPIE games. }
            \small{
            \begin{algorithmic}[1]
                \FOR {each of the indeterminate variable ($i=1$ to
                $\mathcal{K}^+$)}
                    \STATE Apply the Buchberger's Algorithm with the lexicographical order $(x_i \prec x_j), \forall j \neq i$
                    and compute a univariate polynomial in the Gr\"{o}bner bases of the $\mathcal{GS}$.
                    \STATE For the univariate polynomial produced in Step 2, check its irreducibility over $\mathbb{Q}$.
                    \IF {the univariate polynomial has linear factors over $\mathbb{Q}$ }
                        \STATE Substitute each root in the triangular form of the Gr\"{o}bner basis and compute a complete solution tuple corresponding
                        to the root.
                        \IF {the solution tuple verifies to be a Nash equilibrium of the input game}
                            \STATE Declare the input game a non-member to the class of IPIE games and stop.
                       \ENDIF
                   \ENDIF
                \ENDFOR
            \STATE Declare the input game a member to the class of IPIE games.
            \end{algorithmic}
            }
        \end{algorithm}
    }
            \noindent Note that the condition of checking irreducibility
            of each univariate polynomial over $\mathbb{Q}$ rather than $\mathbb{Z}$
            lets us use field algebra, providing a richer set of tools to work with.\\

            \noindent Next, we analyze
            various possibilities of a root generation
            during the process of factorization of the univariate polynomials
            in $\mathcal{GS}$.\\

            \noindent There are mainly five possibilities:
            (1). The first univariate polynomial in the Gr\"{o}bner basis of $\mathcal{GS}$ has
            complete linear factorization over $\mathbb{Q}$. In this situation the game is a
            non-member. (2). The first univariate polynomial has some linear factors over $\mathbb{Q}$.
            Further, substitution of all these linear factors in the triangular form of the Gr\"{o}bner basis
            produce univariate polynomials (in other indeterminate variables) with some linear factors.
            The substitution process generates solutions tuples with
            rational and irrational coordinates.
            There are two possibilities for each of these solution
            tuples. Either the solution is a Nash equilibrium or it is not. In
            former case the game is a
            non-member while in later it may be a member.
            (3). The first univariate polynomial has some linear
            factors but the subsequent univariate polynomials have all the irreducible factors over $\mathbb{Q}$.
            This case is subsumed in case (2) mentioned above.
            (4).The first univariate polynomial has all irreducible factors over $\mathbb{Q}$,
            and subsequent substitutions produce polynomials with rational or integer
            roots. This case is also subsumed in case (2). Finally,
            (5). The first univariate polynomial has all irreducible factors over $\mathbb{Q}$,
            and subsequent substitutions produce all the roots that extend $\mathbb{Q}$.
            In this situation, due to existence of at least one mixed Nash equilibrium \cite{jfn50}, we are
            guaranteed to get at least one irrational solution of the $\mathcal{GS}$. And
            so the game is a member.\\

            \noindent The analysis above, of roots generation during the factorization in
            Algorithm \ref{algorithm:membership}, suggests that
            a solution tuple with rational coordinates must be verified to be a Nash
            equilibrium. In this situation if it is guaranteed that substitution of an irrational root
            produce all irrational roots in subsequent substitution,
            then
            the repeated factorization and verification of the roots can reduce.
            We call this condition \textit{the membership condition}.
            The condition will also improve efficiency of the membership algorithm.\\

            \noindent In the following section, we present a result that
            improves running time of the membership
            Algorithm \ref{algorithm:membership}. The result primarily focuses on
            an algebraic property of the ideal $\mathcal{I}$ of $\mathcal{GS}$.

    \subsection{Result}
    {
            \noindent In this section we present a result that
            allows us to decide membership to the classes of games with
            relatively few root
            computation, factorization and verifications.\\

            \begin{proposition}\label{result:membership lemma}
                If the polynomial ideal $\mathcal{I}$ of the game polynomial $\mathcal{GS}$
                is zero-dimensional, radical and in general position. Then for deciding
                membership to the class of IPIE games exactly one
                irreducible univariate polynomial in the Gr\"{o}bner basis of the $\mathcal{GS}$ is sufficient.
            \end{proposition}
            \begin{proof}
                If the $\mathcal{GS}$ over a field $\mathbb{Z}$ has following form
                \begin{eqnarray}\label{eqn: shape form}
                    x_1 - h_1(x_p) = 0 \nonumber \\
                    x_2 - h_2(x_p) = 0 \nonumber \\
                    ...                \nonumber \\
                    x_{p-1} - h_{p-1}(x_p) = 0\nonumber \\
                    h_p(x_p) = 0.
                \end{eqnarray}
                    And if $h_p$ is irreducible over $\mathbb{Z}$ with $1< deg\ h_i <
                    deg\ h_p,\ 1\leq i < p$. Then, for all $i$, root value of $x_i$ will extend $\mathbb{Z}$.
                    This is a sufficient condition for the membership condition.\\

                    \noindent One of the ways to obtain the condition above for $\mathcal{GS}$ defined over
                    some field $\mathbb{F} \subset \mathbb{Z}$ is the Shape Lemma
                    \cite{GG98}. The shape lemma states:
                 Let the ideal $\mathcal{I}$ of a polynomial system be a zero-dimensional
                 ideal in $\mathbb{F}[x_1,\ldots,x_n]$
                 which is in general position with respect to $x_1$, i.e.
                 the projection of $\mathcal{V}_{\mathbb{K}}(\mathcal{I})$
                 onto the 1-st coordinate is injective. Then $\sqrt
                 \mathcal{I}$ has a lexicographical reduced Gr\"{o}bner basis with
                 respect to $x_n \prec \ldots \prec x_1$ of the form:
                 \[
                     \sqrt \mathcal{I} = \langle g_n(x_n), x_{n-1}-g_{n-1}(x_n),\ldots,x_2-g_2(x_n),x_1-g_1(x_n) \rangle
                 \]
                 where $g_n$ is a square-free polynomial and the degree of
                 every $g_i$ doesn't exceed degree $d$ of $g_n$. $\mathbb{K}$
                 is algebraic closure of $\mathbb{F}$.\\

            \noindent The requirement in the shape lemma gives us
            the required condition, and so the
            result.\hfill\qed
                \end{proof}
                \noindent The result above throws more light on the structure
                of a game system $\mathcal{GS}$ for the class of
                games that we consider. Moreover, it improves
                the computational complexity of Algorithm
                \ref{algorithm:membership}.\\

        \noindent Next we present the correctness of the membership
        Algorithm \ref{algorithm:membership}.
        \begin{proposition}
            The Algorithm \ref{algorithm:membership} for deciding
            membership to the class of IPIE games works.
        \end{proposition}
        \begin{proof}
            \noindent Input to the Algorithm \ref{algorithm:membership} is a finite
            normal form game characterized as $\mathcal{GS}$ of the
            form \eqref{game_polynomial_0}.\footnote{We recall that the type of
            games that we consider, in this work, are known to have finitely
            many equilibria solutions.} All the coefficients of
            the $\mathcal{GS}$ must be integral, otherwise the game
            is a non-member to the class of IPIE games.\\

            \noindent In case all the payoff values are integer,
            then the for loop of Steps 1-10 executes. The other
            property,
            of the input game, that we have to check is, whether
            all its solutions are irrational or not. For checking
            the irrationality
            of each
            indeterminate variable in the $\mathcal{GS}$, we must
            check
            whether its univariate polynomial has all its factors irreducible over $\mathbb{Q}$ or not. To
            obtain a univariate polynomial in each indeterminate
            variable, with
            different lexicographical orders, the Buchberger's algorithm is called. Due to the finiteness of
            number of equilibria solutions of the input game and
            Buchberger \cite{bb01} we are guaranteed to get a univariate
            polynomial every time.\\

            \noindent For checking irreducibility of each of the
            univariate polynomials over $\mathbb{Q}$, algorithm from
            \cite{GG03} is used. The polynomial time irreducibility
            check algorithm takes Algorithm
            \ref{algorithm:membership} to Step 4 every time.\\

            \noindent For each of the linear factors of a
            univariate polynomial a solution tuple is constructed.
            These solutions are then verified to be Nash equilibrium
            of the input game. If one of these
            solutions is a Nash equilibrium, then the method stops
            with declaring the input game a non-member to the class of IPIE
            games.\\

            \noindent With the finiteness of the following: degree bound of the degrees of the
            $\mathcal{GS}$, the degrees of each univariate polynomial and number of strategies,
            we are guaranteed to reach Step 10 of the Algorithm \ref{algorithm:membership}.\\

            \noindent At the end of the for loop, univariate
            polynomials of all the indeterminate variables have all
            non-rational factors and so the input game is declared a
            member to the class of IPIE games.\hfill\qed
        \end{proof}
        \noindent Note that in the light of Proposition \ref{result:membership
        lemma}, the for loop in the Algorithm \ref{algorithm:membership} must be run exactly once.
        And so Algorithm \ref{algorithm:membership} works with the condition in Proposition \ref{result:membership
        lemma}.\\

        \noindent Running time of the Algorithm \ref{algorithm:membership} is primarily dominated by the
        preprocessing task of constructing univariate polynomials.
        The Buchberger's algorithm for constructing univariate polynomials takes doubly exponential time in
        the number of indeterminate variables $\mathcal{K}^+$.\\

        \noindent Our game system is of finite
        size in terms of degree, number of indeterminate variables and its
        norm. And so, keeping aside running time of the Buchberger's
        algorithm, the membership Algorithm \ref{algorithm:membership}
        runs in polynomial time. Further,
        Proposition
        \ref{result:membership lemma} improves running time efficiency of
        the naive algorithm.\\

        \noindent It is important
        to note that, the method for deciding membership does not assume that the Galois groups are
        known. Also, for deciding the membership, the algorithm \ref{algorithm:membership} does not
        compute all the solutions of the $\mathcal{GS}$.

    }

    \section{Conclusion}\label{section_conclusion}
    {

        In this article, we considered the problem of computing all
        the equilibria of a subclass of finite normal form. Our intention
        was to use knowledge of a sample equilibrium for computing all the
        other equilibria. By defining the class of IPIE games, and presenting
        an algorithm for computing all its equilibria, we have addressed
        the problem partially. For computing a
        sample solution we used MVNRM with the KLL algorithm.
        Though MVNRM is not globally convergent but offers significant speed.
        For convergence guarantee we derived a condition.
        With the use of KLL algorithm we computed equilibria in algebraic form
        contrary to traditional approach of keeping solutions in approximation
        form. We also factored out unwanted roots during the sample
        solutions generation, keeping the MVNRM away from local
        minima.
       The sample solution computation algorithm, with the convergence
        condition, is a definitive method for computing exact equilibria
        rapidly.\\

        \noindent Further, for computing all the other equilibria we used
        knowledge of Galois groups. The assumption of
        known Galois groups can easily be relaxed by adjoining an algorithm
        for computing Galois groups. Construction of minimal
        polynomials using the KLL Algorithm and the use of
        Tschruhaus transformation over minimal polynomials help relax
        the criteria.\\

                    \noindent Algorithm that we suggest for solving system of equation
            may not be time efficient for large problems in
            practice, but it is
            time efficient compared to similar method
            based only on Gr\"{o}bner bases.  Algebraic approach that we consider
            in this paper throws more light on
            structure of solutions of class of games.\\

            \noindent Based on the algebraic model in the Section 2,
            we presented an outline of a method for deciding membership to
            the class of IPIE games. We also presented a result for
            improving efficiency of the naive method, and presented the
            correctness of the membership algorithm. \\

            \noindent The complimentary work to the presented work would be
            to prove a result concerning the existence
            of IPIE games in general.\\

            \noindent For the class of IPIE games, the algorithms we suggest are new and
            should be considered as precursor to efficient algorithms in
            future.\\

    \section{Appendix: Example}\label{example: integer games}
    {

        In this section we present an example to show working of the
        algorithms presented. The input 3 players 2 strategy game
        is given by the following payoff matrix. The game is defined
        in \cite{NCH03}.
        \begin{table}[h]
            \centering
            \small{
            \begin{tabular}{c r|c|c|}
                            &\multicolumn{1}{c}{ } & \multicolumn{1}{c}{\textbf{A}} &\multicolumn{1}{c}{\textbf{B}} \\\cline{3-4}
                            &\textbf{a} &\ 3, 0, 2\ &\ 0, 2, 0\ \\\cline{2-4}
                \textbf{1} &\textbf{b} &\ 0, 1, 0\ &\ 1, 0, 0\ \\\cline{1-4}
                \textbf{2} &\textbf{a} &\ 1, 0, 0\ &\ 0, 1, 0\ \\\cline{2-4}
                            &\textbf{b} &\ 0, 3, 0\ &\ 2, 0, 3\
                            \\\cline{2-4}\
            \end{tabular}
                        \caption{Payoff matrix of a 3-player 2-strategy IPIE game. Player 1 and  2's
                        strategies are indicated by a, b and A, B respectively. Player 3's
                        strategies are 1 and 2. Each entry in the matrix indicates player 1, 2
                        and 3's payoff for their respective strategies.
                        \label{tab:PayoffMatrixFor3Player2LinkGame}}
                        }
        \end{table}\\
        \noindent For $i \in \{1,2\}$, we denote $x_i, y_i$ and $z_i$ as
        player 1, 2 and 3's $i^{th}$ strategy, respectively. First, we characterize all the Nash
        equilibria of the game in Table
        \ref{tab:PayoffMatrixFor3Player2LinkGame}
         as solutions of the $\mathcal{GS}$.\\

        \noindent First we decide membership of the input game to class of IPIE games.
        After the characterization of all the Nash
        equilibria of the game above, we apply the Buchberger's algorithm
        to a generate univariate polynomial in a Gr\"{o}bner basis with lexicographical order
        $z_1 \prec x_1 \prec y_1$ and get,
        \begin{equation}\label{gro_in_y}
            y_1^4 + 7y_1^3 -11y_1^2 +3y_1 =0.
        \end{equation}
        \noindent The polynomial has $y_1(y_1-1)(y_1^2 + 8y_1-3)$ as its factorization over
        $\mathbb{Q}$. Substituting rational roots in
        the triangular form of the Gr\"{o}bner basis we get $(0,0,0)$ and
        $(1,1,1)$ as its two solutions with rational coordinates. Verification of these
        solutions as Nash equilibrium of the input game reveals that neither constitute an equilibrium of
        the game. We can repeat the procedure above for $x_1$ and $z_1$ with different lexicographical
        order.\footnote{We do not require to repeat the procedure for $x_2,y_2$ and $z_2$ because its a
        two strategy game.}. We see
        that the ideal of the game polynomial for this game follows
        the shape lemma and so the irreducible factor of indeterminate $y_1$
        guarantees that the game is a member to the class of IPIE
        game.\\

       \noindent Next we apply equilibria computation algorithm for the game above and
        compute its equilibria using group action. With the initial guess of the solution tuple
        consisting of all the 0's, $d=2$ and $H=3$
         we apply MVNRM and compute an approximate sample solution
         tuple as follow.
        \begin{equation}\label{gro_in_x}
            x_1 := 0.7282202113;\quad y_1 := 0.3588989435;\quad z_1 := 0.4717797888
        \end{equation}
        \noindent Applying the KLL
        algorithm on the solution above generates univariate polynomials as follows.
        \begin{equation}\label{factor_y}
            5x_1^2 - 16x_1 + 9 =0 ;\quad  y_1^2 + 8y_1 - 3 =0;\quad 5z_1^2 + 4z_1
            -3 =0
        \end{equation}
        \noindent These polynomials are irreducible over
        $\mathbb{Z}$ and has a Galois group \{id,conjugate\}, isomorphic to
        $\mathbb{Z}_2$.
        With minimal polynomials, we compute a solution of the $\mathcal{GS}$ in closed form.
        Let one such solution be,
        \begin{equation}\label{example_solution}
            x =\frac{1}{5}(8+\sqrt 19); y=-4-\sqrt 19;z = \frac{1}{5}(-2-\sqrt
            19).
        \end{equation}
        This is a sample solution of the game system.
        Next we perform Galois group action on the sample solution.
        Once all the solutions are computed, we reject non-equilibria
        solution of the game with the polynomial time verification algorithm
        \cite{rg05}. This gives us
        the unique irrational equilibrium of the game depicted in Table \ref{tab:PayoffMatrixFor3Player2LinkGame},
        \begin{equation}\label{example_solution}
            x =\frac{1}{5}(8-\sqrt 19); y=-4+\sqrt 19;z = \frac{1}{5}(-2+\sqrt
            19).
        \end{equation}

    \bibliographystyle{splncs}
    \bibliography{references}

\end{document}